\documentclass[a4paper,UKenglish,dvipsnames]{lipics-v2016}
 
\usepackage{microtype}

\usepackage[dvipsnames]{xcolor}
\usepackage[fleqn]{amsmath}
\usepackage{amssymb,latexsym,wasysym,dsfont,stmaryrd}
\usepackage{amsthm}
\usepackage{relsize}
\usepackage{xspace}
\usepackage{graphicx}
\usepackage{gastex}
\usepackage{tikz}
\usetikzlibrary{arrows,automata,shapes.geometric,decorations.pathmorphing,backgrounds,positioning,fit,petri,calc}
\usepackage{wrapfig}
\usepackage{macros}
\usepackage{manfnt}              
\usepackage{enumitem}
\setlist[enumerate]{itemsep=0mm}
\setlist[itemize]{itemsep=0mm,topsep=1pt,partopsep=0pt}

\newif\ifdraft\draftfalse

\ifdraft
\newcommand\modedraft[1]{#1}
\newcommand\todo[1]{{\color{purple}[\textbf{To do:} #1]}}
\newcommand\bmcomment[1]{\marginpar[{\color{blue}\small\dbend}]{\color{blue}\small\dbend}{\footnotesize \color{blue}[#1 - \textbf{Bastien}]}}
\newcommand\amcomment[1]{\marginpar[{\color{OliveGreen}\small\dbend}]{\color{OliveGreen}\small\dbend}{\footnotesize
    \color{OliveGreen}[#1 - \textbf{Nello}]}}
\newcommand\rbcomment[1]{\marginpar[{\color{Orange}\small\dbend}]{\color{Orange}\small\dbend}{\footnotesize \color{Orange}[#1 - \textbf{Raphael}]}}

\else
\newcommand\modedraft[1]{}
\newcommand\todo[1]{}
\newcommand\bmcomment[1]{}
\newcommand\amcomment[1]{}
\newcommand\rbcomment[1]{}

\fi

\title{Quantified CTL with imperfect information\footnote{This project
    has received funding from the European Union's Horizon 2020
    research and innovation programme under the Marie Sklodowska-Curie
    grant agreement No 709188.}}

\newcommand\UElogo{%
\begin{tikzpicture}[remember picture,overlay]
\node[anchor=south,yshift=4.2cm,xshift=2cm] at (current page.south) {\includegraphics[height=2.5em]{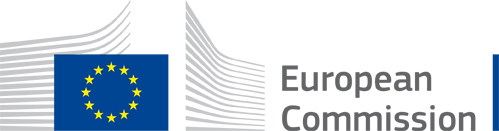}};
\end{tikzpicture}%
}

\author[1]{Rapha\"el Berthon}
\author[2]{Bastien Maubert}
\author[3]{Aniello Murano}
\affil[1]{\'Ecole Normale Supérieure de Rennes, France\\
  \texttt{raphael.berthon@ens-rennes.fr}}
\affil[2]{Universit\`a degli Studi di Napoli Federico II, Italy\\
  \texttt{bastien.maubert@gmail.com}}
\affil[3]{Universit\`a degli Studi di Napoli Federico II, Italy\\
  \texttt{murano@na.infn.it}}

\authorrunning{R. Berthon, B. Maubert and A. Murano} 

\Copyright{Rapha\"el Berthon, Bastien Maubert and Aniello Murano}


\theoremstyle{plain}

\begin{document}
\maketitle
\UElogo

\begin{abstract}
  Quantified \CTL (\QCTL) is a well-studied temporal logic
 that extends \CTL with quantification over
  atomic propositions. It has recently come to the fore as a
  powerful intermediary framework to study logics
  for  strategic reasoning.
We extend it to include imperfect information  
by parameterising  quantifiers with an observation that
  defines how well they observe the model, thus  constraining their behaviour. 
%
We consider two different semantics, one related to the
notion of \emph{no memory}, the other to \emph{perfect recall}.  
We study the expressiveness of our logic, and show that it coincides
with \MSO for the first semantics and with \MSO with equal level for
the second one.
We establish that the model-checking problem is \PSPACE-complete
for the first semantics. While it is undecidable for the second one, we
identify a  syntactic fragment, defined by a notion of
hierarchical formula, which we prove to be decidable thanks to an
automata-theoretic approach. 

\end{abstract}

\begin{section}{Introduction}

  Temporal logic is a powerful framework widely used in formal
  system-design and verification~\cite{CGP02,Pnu77}.  It allows
  reasoning over the temporal evolution of a system, without referring
  explicitly to the elapsing of time. One of the most significant
  contributions of the field is \emph{model checking}, which allows to verify
  system correctness by checking whether a mathematical model of the
  system satisfies a temporal logic formula expressing its desired
  behaviour~\cite{CE81,CGP02,KVW00,KVW01}.
  %
  \\ \indent
  Depending on the view of the  nature of time, two types of
  temporal logics are mainly considered.  In \emph{linear-time
    temporal logics} such as {\LTL}~\cite{Pnu77} time is treated as
  if each moment in time had a unique possible future. Conversely, in
  \emph{branching-time temporal logics} such as \CTL~\cite{CE81} and
  \CTLs~\cite{EH86}, each moment in time may split into various
  possible futures;  existential and universal
  quantifiers then allow expressing properties of either one or all the
  possible futures. While {\LTL} is suitable to express path properties, \CTL is
  more appropriate for state-based ones, and \CTLs for both.  These
  logics are ``easy-to-use'', can express important system
  properties such as \emph{liveness} or \emph{safety}, enjoy
good  fundamental theoretical properties such as invariance under
  tree-unwinding of models, and come with reasonable complexities for the
  main related decision problems.  For instance, the model-checking and
  satisfiability problems for \CTLs are \PSPACE-Complete~\cite{EL85} and
  \2EXPTIME-Complete~\cite{VS85}, respectively.
  %
  \\ \indent
  Along the years, \CTLs has been extended in a number of ways in
  order to verify the behavior of a broad variety of
  systems.  In multi-agent open-system verification, \emph{Alternating-Time Temporal Logic} (\ATLs),
  introduced by Alur, Henzinger, and Kupferman~\cite{AHK02}, is
  particularly successful.  This
  generalization of \CTLs  replaces path
  quantifiers with \emph{strategic modalities}, that is modalities
  over teams of agents that describe the ability to cooperate  in
  order to achieve a goal against adversaries.
  \ATLs model checking is a very active area of research and it has been
  studied in several domains, including communication
  protocols~\cite{Hoek03ATELstudialogica}, fair exchange
  protocols~\cite{Kremer03fairexchange,Jamroga12fairness-lncs}, and
  agent-oriented programs~\cite{Dastani10programs-aamas}.
  The complexity of the problem has been extensively studied in a
  multitude  of papers, and algorithms have been implemented in
  tools~\cite{Lomuscio06mcmas}.
 Remarkably, \ATLs has inspired fresh and more powerful logics such as
 \emph{Strategy Logic} (\SL)~\cite{CHP10,MMV10b,MMPV14},  \emph{\ATLs with
strategy context} (\ATLSsc)~\cite{BLLM09,DLM10},
 \emph{\ATLs with Irrevocable
 strategies} (I\ATLs)~\cite{AGJ07} and \emph{Memoryful \ATLs}
 (m\ATLs)~\cite{MMV10a}.
 These logics are progressively overtaking \ATLs; in
 particular this is the case for \SL
 as it
 can can express fundamental game-theoretic concepts such as
 Nash Equilibrium and Subgame Perfect Equilibrium~\cite{MMPV14}.
  \\ \indent
  In the landscape of temporal logics, another breakthrough
  contribution comes from \emph{Quantified \CTLs} (\QCTLs), which
  extends \CTLs with the possibility to quantify over atomic
  propositions~\cite{Sis83,Kup95,KMTV00,french2001decidability,DBLP:journals/corr/LaroussinieM14}.  \QCTLs turns out
  to be very expressive (indeed, it is equivalent to Monadic
  Second-Order Logic, \MSO for short) and was usefully
  applied in a number of scenarios. Recently it has come to the fore
  as a convenient and uniform intermediary logic to easily obtain 
  algorithms for \ATLSsc, \SL, as well as related
  formalisms~\cite{DBLP:journals/corr/LaroussinieM14,LLM10,MMPV14}. Indeed, strategies 
   can be represented by  atomic propositions
  labelling the execution tree of the game structure under study, and
strategy quantification can thus be expressed by means of propositional
  quantifications. As
  a remark, quantification in \QCTLs can be interpreted either on Kripke structures
  (\emph{structure semantics}) or their execution tree (\emph{tree
    semantics}), allowing for the encoding of memoryless or perfect-recall strategies, respectively. This difference impacts also the complexity of the
  related decision problems: for instance, moving from structure to tree
  semantics, model checking jumps from \PSPACE to non-elementary. 
  %
  \\ \indent
  In game theory and open-system verification 
  an important body of work has been devoted to \emph{imperfect
    information}, which refers to settings in which players have
  partial information about the moves taken by the
  others~\cite{CLMM14,DT11,JM14,KV97,PR89}.  This is a common
  situation in real-life scenarios where 
 players 
 have to act without having all the relevant information at hand. In
  computer science this situation occurs for example when some system's
  variables are internal/private~\cite{Reif84}.  
  Imperfect information is usually modelled by 
  indistinguishability relations over the states of the
  game. During a play,  some players may
  not know precisely in which state they are, and therefore they cannot base their
  actions on the exact current situation: they must
  choose their actions uniformly over
  indistinguishable states~\cite{KV97}.
  %
  \\ \indent
  This uniformity constraint 
 deeply impacts  the complexity of decision
  problems. It is well known that multi-player games with
  imperfect information are   computationally hard, in general 
 undecidable~\cite{peterson2001lower}, 
  and to retain positive complexity results one needs to
  restrict players' capabilities, by bounding their memory of
  past moves~\cite{CT11} or putting some hierarchical order over
  their observational power~\cite{PR89}. Unfortunately, most of the 
  approaches exploited under full observation  are not appropriate for
  imperfect information. In  particular this is the case of \QCTLs,
  unless opportunely adapted. In this paper we work in this direction
  by incorporating in \QCTLs the essence of imperfect
  information, that is the uniformity constraint on choices. We believe
  it may provide a uniform
  framework to obtain new results on logics for strategic
  reasoning under imperfect information, as does \QCTLs
  in   the perfect information setting.
  \vspace{0.20truecm}
  \\ \indent
  {\bf Our contribution.}  We introduce \QCTLsi, an opportune
  extension of \QCTLs that integrates the central feature of imperfect
  information, \ie, uniformity constraint on choices.  
We add internal structure to the states of the models, much like in
  Reif's multiplayer
  game structures~\cite{peterson2001lower} or distributed
  systems~\cite{halpern1989complexity}, and
  we parameterise
  propositional quantifiers   with observations that define what
  portions of the  states a quantifier can
  ``observe''. The semantics is adapted to capture the idea
  of quantifications on atomic propositions being made with  partial
  observation. Like in~\cite{DBLP:journals/corr/LaroussinieM14}, we consider both
  structure and tree semantics. 
  \\ \indent
We study the expressive power of \QCTLsi. By using the same argument as for
\QCTLs~\cite{DBLP:journals/corr/LaroussinieM14}, we first show that \QCTLsi and \QCTLi are equally
expressive for both semantics. Then we
prove that for the structure semantics, these logics are no more expressive
than \QCTL, and thus coincide with \MSO. Finally we  show that
under tree semantics \QCTLi is expressively equivalent to \MSO extended with
the equal level predicate (\MSOeql, see \cite{elgot-rabin66,lauchli1987monadic,thomas-msoeqlevel}).
  \\ \indent
Concerning the model-checking problem  we first prove
that under structure semantics it is \PSPACE-complete for both \QCTLsi
and \QCTLi, like \QCTL.
Under tree semantics, undecidability follows from the
equivalence with \MSOeql. However we identify a decidable syntactic
fragment, consisting of those formulas in which nested quantifiers
 have hierarchically ordered observations, innermost ones
observing more than outermost ones. We call such formulas
\emph{hierarchical formulas}. Interestingly, a decidability result for
Quantified $\mu$-Calculus with partial observation
\cite{pinchinat2005decidable} uses a similar syntactic
restriction. This logic is very close to ours, but orthogonal:
 while our tree
semantics relies on a synchronous perfect-recall notion of imperfect
information, theirs is asynchronous.
This hierarchical restriction is also related to  decidability results
for games with imperfect information \cite{peterson2002decision,DBLP:conf/atva/BerwangerMB15} and
distributed synthesis~\cite{kupermann2001synthesizing}.
Our decision procedure relies on automata constructions involving the
\emph{narrowing} operation introduced by Kupferman and Vardi in
\cite{kupferman1999church} for distributed synthesis. We believe that
our choice of modelling imperfect information by means of local
states 
eases greatly the use of automata techniques to tackle imperfect information.
Finally, our result provides new decidability results for \ATLSsc
with imperfect information (not presented here), and we trust it will
find applications in other logics, such as \SL with imperfect
information.
  \\ \indent
  {\bf Plan.} In Section~\ref{sec-prelim} we recall Kripke structures and trees,
  and the syntax and semantics of \QCTLs. We then present \QCTLsi in Section~\ref{sec-QCTL-imp-inf}, we study
  its expressiveness in Section~\ref{sec-expressivity} and its
  model-checking problem in Section~\ref{sec-modelchecking}. We
  conclude and discuss future work in Section~\ref{sec-conclusion}.



%
%
%
%

%
%

\end{section}


\section{Preliminaries}
\label{sec-prelim}
Let $\Sigma$ be an alphabet. A \emph{finite} (resp. \emph{infinite}) \emph{word} over $\Sigma$ is an element
of $\Sigma^{*}$ (resp. $\Sigma^{\omega}$). The empty word is
classically noted $\epsilon$, and
$\Sigma^{+}=\Sigma^{*}\setminus\{\epsilon\}$. The \emph{length} of a
  word is $|w|\egdef 0$ if $w$ is the empty word $\epsilon$, if $w=w_{0}w_{1}\ldots
w_{n}$ is a finite non-empty word then $|w|\egdef n+1$, and for an infinite word $w$ we let $|w|\egdef
\omega$. Given a word $w$ and $0\leq i,j\leq |w|-1$, we let $w_{i}$ be the
letter at position $i$ in $w$ and $w[i,j]$ be the subword of $w$ that starts
at position $i$ and ends at position $j$.
If $w$ is infinite, we let $w^{i}\egdef w[i,\omega]$.
We write $w\pref w'$ if $w$ is a prefix of $w'$, and $\FPref{w}$ is
the set of finite prefixes of word $w$.
Finally, for $n\in\setn$ we let $[n]\egdef\{1,\ldots,n\}$.



\subsection{Kripke structures and trees}

Let $\AP$ be a countably infinite set of
 \emph{atomic propositions} and let $\APf\subset\AP$ be a finite subset.
 
\begin{definition}
  A \emph{Kripke structure} over $\APf$ is a tuple
  $\KS=(\setstates,\relation,\lab)$ where
  $\setstates$
  is a set
  of \emph{states},
  $\relation\subseteq\setstates\times\setstates$ is a
  left-total\footnote{\ie, for all $\state\in\setstates$, there exists
    $\state'$ such that $(\state,\state')\in\relation$.}
  \emph{transition relation} and $\lab:\setstates\to 2^{\APf}$ is a
  \emph{\labeling function}. 
\end{definition}

A \emph{pointed Kripke structure} is a
  pair $(\KS,\state)$ where $\state\in\KS$, and the
  \emph{size} $|\KS|$ of a Kripke structure $\KS$ is its number of states.
A \emph{path} in a structure   $\KS=(\setstates,\relation,\lab)$  is
an infinite word $\spath\in\setstates^{\omega}$
 such that for all $i\in\setn$,
$(\spath_{i},\spath_{i+1})\in \relation$. For 
$\state\in\setstates$, we let $\Paths(\state)$ be the set of all
paths that start in $\state$.  
A \emph{finite path} is a finite non-empty prefix of a path.


We now define (infinite) trees. In many works, trees are defined as prefixed-closed sets of
words
with the empty word $\epsilon$ as root. Here trees  represent unfoldings of
 Kripke
structures, and we find it more convenient to see a node as a sequence of
states and the root as the initial state, hence the
following definition, where $\Dirtree$ is a finite set: 

\begin{definition}
\label{def-tree}
An \emph{$\Dirtree$-tree} $\tree$ 
 is a
nonempty set of words $\tree\subseteq \Dirtree^+$ such that:
\begin{itemize}
  \item\label{p-root} there exists $\racine\in\Dirtree$,  called the
    \emph{root} of $\tree$, such that each
    $\noeud\in\tree$ starts with $\racine$;
  \item if $\noeud\cdot\dir\in\tree$ and $\noeud\neq\epsilon$, then
    $\noeud\in\tree$, and
  \item if $\noeud\in\tree$ then there exists $\dir\in\Dirtree$ such that $\noeud\cdot\dir\in\tree$.
\end{itemize}
\end{definition} 

The elements of a tree $\tree$ are called \emph{nodes}.  
  If 
 $\noeud\cdot\dir \in \tree$, we say that $\noeud\cdot\dir$ is a \emph{child} of
 $\noeud$. An $\Dirtree$-tree is \emph{full} if every node $\noeud$ has a child
 $\noeud\cdot\dir$ for each $\dir\in\Dirtree$.
 The \emph{depth} of a node $\noeud$ is $|\noeud|$.
Similarly to Kripke structures, a \emph{\tpath} is an infinite sequence of nodes $\tpath=\noeud_0\noeud_1\ldots$
such that for all $i\in\setn$, $\noeud_{i+1}$ is a child of
$\noeud_i$,
and $\tPaths(\noeud)$ is the set of \tpaths
 that start in node $\noeud$. 
An \emph{$\APf$-\labeled $\Dirtree$-tree}, or
\emph{$(\APf,\Dirtree)$-tree} for short, is a pair
$\ltree=(\tree,\lab)$, where $\tree$ is an $\Dirtree$-tree called the
\emph{domain} of $\ltree$ and
$\lab:\tree \rightarrow 2^{\APf}$ is a \emph{\labeling}.
For a labelled tree $\ltree=(\tree,\lab)$ and an atomic
proposition $p\in\AP$, we define the \emph{$p$-projection of $\ltree$}
as the labelled tree
$\proj{\ltree}\;\egdef(\tree,\proj{\lab})$, where for each
$\noeud\in\tree$, $\proj{\lab}\!(\noeud)\egdef \lab(\noeud)\setminus
\{p\}$. For a set of trees $\lang$, we let $\proj{\lang}\;\egdef\{\proj{\ltree}\;\mid\ltree\in\lang\}$.

\begin{definition}[Tree unfoldings]
  \label{sec-unfoldings}
  Let $\CKS=(\setstates,\relation,\lab)$ be a Kripke structure over $\APf$, and let $\state\in\setstates$. 
  The \emph{tree-unfolding of $\CKS$ from $\state$} is the 
  $(\APf,\setstates)$-tree $\unfold{\state}=(\tree,\lab')$, where
    $\tree$ is the set
    of all finite  paths that start in $\state$, and
    for every $\noeud\in\tree$,
    $\lab'(\noeud)=\lab(\last(\noeud))$.
\end{definition}


\subsection{\QCTLs, syntax and semantics}
\label{sec-QCTL}

We recall the syntax of \QCTLs, as well as both the structure and tree semantics.
\begin{definition}
  The syntax of \QCTLs is defined by the following grammar:
  \begin{align*}
  \phi\egdef &\; p \mid \neg \phi \mid \phi\ou \phi \mid \E \psi \mid
  \existsp{} \phi\\
    \psi\egdef &\; \phi \mid \neg \psi \mid \psi\ou \psi \mid \X \psi \mid
  \psi \until \psi
\end{align*}
where $p\in\AP$. Formulas of type $\phi$ are called \emph{state
  formulas}, those of type $\psi$ are called \emph{path formulas}, and
 \QCTLs consists of state formulas. 
\end{definition}

Like in \cite{DBLP:journals/corr/LaroussinieM14} we consider two different semantics, the 
\emph{structure semantics} and the \emph{tree
  semantics}: in the former formulas are evaluated directly on the
structure, while in the latter the structure is first unfolded into an
infinite tree. In the first case, quantifying over 
 $p$  means choosing a truth value for $p$ in each state of
the structure, while in the second case it is possible to choose a
different truth value for $p$ in each finite path of the structure. 


\subsubsection{Structure semantics}

A \QCTLs state (resp. path) formula is evaluated in a state
(resp. path) of a Kripke structure.
To define the semantics of quantifications over propositions, the
following definition is handy.

\begin{definition}
  \label{def-pequiv-struct}
  For $p\in\AP$, two structures $\KS=(\setstates,\relation,\lab)$ and
  $\KS'=(\setstates',\relation',\lab')$ are \emph{equivalent modulo $p$},
  written $\KS\Pequiv\KS'$, if $\setstates=\setstates'$, $\relation
  =\relation'$ and for each state $\state\in\setstates$,
  $\lab(\state)\setminus\{p\}=\lab'(\state)\setminus\{p\}$.
  This definition also applies to labelled trees seen as infinite Kripke structures.
\end{definition}

 The satisfaction relation $\modelss$ for the structure semantics is
defined inductively as follows, where   $\KS=(\setstates,\relation,\lab)$ is a Kripke structure, 
$\state$ is a state and $\spath$ is a path in $\KS$:
\[
\begin{array}{lclclcl}
  \CKS,\state\modelss p & \mbox{if} & p\in\lab(\state) & \makebox[\widthof{blablablabla}]{} &
    \CKS,\state\modelss \neg \phi &
    \mbox{if} & \CKS,\state\not\modelss \phi \\ 
        \CKS,\state\modelss  \phi \ou \phi'&
    \mbox{if} & \CKS,\state \modelss \phi \mbox{ or
    }\CKS,\state\modelss \phi' \\
    \CKS,\state\modelss \E\psi &
    \mbox{if} & \multicolumn{5}{l}{\mbox{there exists }\spath\in\Paths(\state) \mbox{ such
      that }\CKS,\spath\modelss \psi} \\
        \CKS,\state\modelss \existsp{} \phi &
    \mbox{if} & \multicolumn{5}{l}{\mbox{there exists }\CKS'\Pequiv[p]\CKS \mbox{ such that
    }\CKS',\state\models\phi}\\
      \CKS,\spath\modelss \phi & \mbox{if} &
      \CKS,\spath_{0}\modelss\phi & \makebox[\widthof{blablablabla}]{} &
    \CKS,\spath\modelss \neg \psi &
    \mbox{if} & \CKS,\spath\not\modelss \psi \\
        \CKS,\spath\modelss  \psi \ou \psi'&
    \mbox{if} & \CKS,\spath \modelss \psi \mbox{ or
    }\CKS,\spath\modelss \psi' & \makebox[\widthof{blablablabla}]{} &
    \CKS,\spath\modelss \X\psi &
    \mbox{if} & \CKS,\spath^{1}\modelss \psi \\
        \CKS,\spath\modelss \psi\until\psi' &
    \mbox{if} & \multicolumn{5}{l}{\mbox{there exists }i\geq 0 \mbox{ such that
    }\CKS,\spath^{i}\modelss\psi' \mbox{ and for }0\leq j <i,\; \CKS,\spath^{j}\modelss\psi}
\end{array}
\]

\subsubsection{Tree semantics}

In the tree semantics, a formula holds in a state $\state$ of a
structure $\KS$
if it holds in the tree-unfolding of $\KS$ from $\state$.
The semantics of \QCTLs on trees could be derived from the structure
semantics, seeing $2^{\APf}$-\labeled trees as  infinite-state Kripke structures.
We define it
explicitly on trees though, as it will make the presentation of the
semantics for \QCTLi clearer.


 The satisfaction relation $\modelst$ for the tree semantics is
thus defined inductively as follows, where   $\ltree=(\tree,\lab)$ is
a $2^{\APf}$-\labeled $\Dirtree$-tree, 
$\noeud$ is a node and $\tpath$ is a path in $\tree$:
\[
\begin{array}{lclclcl}
  \ltree,\noeud\modelst p & \mbox{if} & p\in\lab(\noeud) &
  \makebox[\widthof{blablablablaaa}]{} &
    \ltree,\noeud\modelst \neg \phi &
    \mbox{if} & \ltree,\noeud\not\modelst \phi \\
        \ltree,\noeud\modelst  \phi \ou \phi'&
    \mbox{if} & \multicolumn{5}{l}{\ltree,\noeud \modelst \phi \mbox{ or
    }\ltree,\noeud\modelst \phi'} \\
    \ltree,\noeud\modelst \E\psi &
    \mbox{if} & \multicolumn{5}{l}{\mbox{there exists }\tpath\in\tPaths(\noeud) \mbox{ such
      that }\ltree,\tpath\modelst \psi} \\
        \ltree,\noeud\modelst \existsp{} \phi &
    \mbox{if} & \multicolumn{5}{l}{\mbox{there exists }\ltree'\Pequiv[p]\ltree \mbox{ such that
    }\ltree',\noeud\models\phi}\\
      \ltree,\tpath\modelst \phi & \mbox{if} &
      \ltree,\tpath_{0}\modelst\phi &
      \makebox[\widthof{blablablablaaa}]{} &
    \ltree,\tpath\modelst \neg \psi &
    \mbox{if} & \ltree,\tpath\not\modelst \psi \\
        \ltree,\tpath\modelst  \psi \ou \psi'&
    \mbox{if} & \ltree,\tpath \modelst \psi \mbox{ or
    }\ltree,\tpath\modelst \psi' & \makebox[\widthof{blablablablaaa}]{} &
    \ltree,\tpath\modelst \X\psi &
    \mbox{if} & \ltree,\tpath^{1}\modelst \psi \\
        \ltree,\tpath\modelst \psi\until\psi' &
    \mbox{if} & \multicolumn{5}{l}{\mbox{there exists }i\geq 0 \mbox{ such that
    }\ltree,\tpath^{i}\modelst\psi' \mbox{ and for }0\leq j <i,\; \ltree,\tpath^{j}\modelst\psi}
\end{array}
\]

We may write $\ltree\modelst\phi$ for $\ltree,\racine\modelst\phi$,
where $\racine$ is the root of $\ltree$, and given a Kripke structure $\KS$, a state $\state$  and a
\QCTLs formula $\phi$, we write $\KS,\state\modelst\phi$ if $\unfold[\KS]{\state}\modelst\phi$.


\section{\QCTLs with imperfect information}
\label{sec-QCTL-imp-inf}

We now enrich the models, syntax and semantics to capture
the idea of quantifications on atomic propositions being made with a
partial observation of the system.

\subsection{Compound Kripke structures}

First, we enrich Kripke structures by adding internal structure to
states: we set them as tuples of
local states. To ease presentation and obtain finite alphabets for our
tree automata in Section~\ref{sec-decidable}, we fix a collection
$\{\setlstates_{i}\}_{i\in [n]}$ of $n$ disjoint finite sets
of \emph{local states}. 

For 
$I\subseteq [n]$, we let $\Dirtreei\egdef\bigtimes_{i\in I}\setlstates_{i}$.
Let $J\subseteq I\subseteq [n]$. For $\dir=(\lstate_{i})_{i\in I}\in\Dirtreei[I]$, we define
the \emph{$\Dirtreei[J]$-projection} of $\dir$ as
$\projI[{\Dirtreei[J]}]{\dir}\egdef (\lstate_{i})_{i\in J}$. If $J=\emptyset$, we let $\projI[\emptyset]{\dir}\egdef
\blank$, where $\blank$ is a special symbol, and we let $\Dirtreei[\emptyset]\egdef\{\blank\}$.
This definition extends naturally to words and trees over
$\Dirtreei[I]$.
Observe that when projecting a tree, nodes with same projection are
merged. In particular,  for every $\Dirtreei[I]$-tree $\tree$,
$\projI[\emptyset]{\tree}$ is the only $\Dirtreei[\emptyset]$-tree, $\blank^{\omega}$.
We also define a \emph{lift} operator
$\liftI[I]{\dira}{}$ that, given  an
$\Dirtreei[J]$-tree rooted in $\dir$  and a tuple $\dira\in
\Dirtreei[I\setminus J]$, produces the $\Dirtreei[I]$-tree
rooted in $(\dir,\dira)$ defined as $\liftI[{\Dirtreei[I]}]{\dira}{\tree}\egdef
\{\noeud\in
(\dir,\dira)\cdot\Dirtreei[I]^{*}\mid
\projI[{\Dirtreei[J]}]{\noeud}\in\tree \}$.
Observe that because the sets $\{\setlstates_{i}\}_{i\in [n]}$ are
disjoint, the ordering of elements in tuples of $\Dirtreei[I]$ does not matter.
For an
 $(\APf,\Dirtreei[J])$-tree
$\ltree=(\tree,\lab)$, we define $\liftI[{\Dirtreei[I]}]{\dira}{\ltree}\egdef (\liftI[{\Dirtreei[I]}]{\dira}{\tree},\lab')$
where $\lab'(\noeud)\egdef \lab(\projI[{\Dirtreei[J]}]{\noeud})$. 
In the following we may write $\liftI[I]{\dira}{}$ for
$\liftI[{\Dirtreei[I]}]{\dira}{}$, and $\projI[J]{}$ instead of $\projI[{\Dirtreei[J]}]{}$.

\begin{definition}
A \emph{compound Kripke structure}, or \CKS, is a Kripke structure
  $\CKS=(\setstates,\relation,\lab)$ such that
  $\setstates\subseteq \Dirtreei[{[n]}]$. 
We call  $n$ the \emph{dimension} of \CKSs. 
\end{definition}

\begin{remark}
Note that by fixing finite sets of local states, we also fix a finite set of possible
states.  If it were not so,    our translation from \QCTLi to \QCTL in
Theorem~\ref{th-qctli-qctl}, as well as the one from \QCTLi to \MSOeql
in
Theorem~\ref{th-qctl-mso-eql} would no longer be valid, making us also
lose Corollary~\ref{cor-qctl-mso}.
We would no longer have equivalence in expressivity, but we would
still have that \QCTLi is at least as expressive as \MSO (resp. \MSOeql) for
structure semantics (resp. tree semantics).
Also our results on
model checking in Section~\ref{sec-modelchecking}, and in particular
our main result, Theorem~\ref{theo-decidable}, would still be valid.  
\end{remark}

To model the fact that quantifiers may not observe some local states, 
we define
a notion of observation and the associated notion of observational indistinguishability.

\begin{definition}
  \label{def-obs-struct}
An \emph{observation} is a finite set of
indices $\obs\subset\setn$. For an observation $\obs$ and $I\subseteq
[n]$, 
  two
  tuples $\dir,\dir'\in\Dirtreei[I]$
are \emph{$\obs$-indistinguishable},
written $\dir\oequiv\dir'$, if
$\projI[{I\cap\obs}]{\dir}=\projI[{I\cap\obs}]{\dir'}$.
\end{definition}

Intuitively, a quantifier with observation $\obs$ must choose the
valuation of atomic propositions \emph{uniformly} with respect to
$\obs$,  
and this notion of uniformity will vary
between the structure semantics and the tree semantics.
But first, let us introduce the syntax of \QCTLsi.

\subsection{\QCTLsi, syntax and semantics}

The syntax of \QCTLsi is  that of \QCTLs, except that
quantifiers over atomic propositions are parameterised by a set of
indices that defines what local states the quantifier can
``observe''.

\begin{definition}
  \label{def-syntax-QCTLsi}
  The syntax of \QCTLsi is defined by the following grammar:
  \begin{align*}
  \phi\egdef &\; p \mid \neg \phi \mid \phi\ou \phi \mid \E \psi \mid
  \existsp[p]{\obs} \phi\\
    \psi\egdef &\; \phi \mid \neg \psi \mid \psi\ou \psi \mid \X \psi \mid
  \psi \until \psi
\end{align*}
where $p\in\AP$ 
and $\obs\subset \setn$
is an observation.
\end{definition}

We use standard abbreviations:
$\top\egdef p\ou\neg p$, $\perp\egdef\neg\top$, $\F\psi \egdef \top \until \psi$, $\G\psi \egdef
\neg \F \neg \psi$ and $\A\psi \egdef \neg\E\neg\psi$.
The size $|\phi|$ of a formula
 $\phi$ is defined inductively as usual, but the following case: $|\existsp{\obs}\phi|\egdef 1 + |\obs| + |\phi|$.
We also classically define the syntactic fragment \QCTLi:

\begin{definition}
    \label{def-syntax-QCTLi}
  The syntax of \QCTLi is defined by the following grammar:
  \begin{align*}
  \phi\egdef &\; p \mid \neg \phi \mid \phi\ou \phi \mid 
  \existsp[p]{\obs} \psi \mid \E \X \phi \mid \A \X \phi \mid
\E \phi \until \phi \mid \A \phi \until \phi
\end{align*}
where $p\in\AP$ 
and $\obs\subset \setn$
is an observation.
\end{definition}

\subsubsection{Structure semantics}

In the case of structure semantics, uniformity 
is defined as follows:

\begin{definition}
  \label{def-unif-syst}
  Let $\CKS=(\setstates,\relation,\lab)$ be a \CKS, $p\in\AP$ and $\obs\subset\setn$.
 $\CKS$ is \emph{$\obs$-uniform in $p$} if for every pair of states
 $\state,\state'\in\setstates$ such that $\state\oequiv\state'$, it
 holds that
 $p\in\lab(\state)$ if and only if $p\in\lab(\state')$.
\end{definition}

We enrich the satisfaction relation
$\modelss$ with the following inductive case, where $(\CKS,\state)$ is
a pointed \CKS:
\[
\begin{array}{lcl}
  \CKS,\state\modelss \existsp{\obs} \phi &
  \mbox{if} & \mbox{there exists }\CKS'\Pequiv[p]\CKS \mbox{ such that
  }\CKS'\mbox{ is $\obs$-uniform in $p$ and }\CKS',\state\modelss\phi
\end{array}
\]

Observe that 
$\existsp{\{1,\ldots,n\}}\phi$ is equivalent to the \QCTLs formula 
 $\existsp{}\phi$. 

\subsubsection{Tree semantics}
\label{sec-QCTLi-tree-sem}

As  observed in the introduction, propositional quantifiers 
 can be seen as having perfect recall in the tree semantics and no
 memory in  the structure semantics.  The following definition for 
 indistinguishability on trees, which differs from that for \CKS, reflects this difference.

 \begin{definition}
   \label{def-obs-trees}
        Let $\ltree=(\tree,\lab)$ be a labelled $\Dirtreei[I]$-tree,
        $p\in\AP$ an atomic proposition and $\obs\subset\setn$ an observation. 
   Two nodes
   $\noeud=\noeud_{0}\ldots\noeud_{i}$ and
   $\noeud'=\noeud'_{0}\ldots\noeud'_{j}$ of $\tree$ are
   \emph{$\obs$-indistinguishable}, written $\noeud\oequivt\noeud'$, if
   $i=j$ and for all $k\in \{0,\ldots,i\}$ we have
   $\noeud_{k}\oequiv\noeud'_{k}$.
Tree $\ltree$ is \emph{$\obs$-uniform in $p$} if for every pair of nodes
 $\noeud,\noeud'\in\tree$ such that $\noeud\oequivt\noeud'$, we have
 $p\in\lab(\noeud)$ iff $p\in\lab(\noeud')$.
 \end{definition}
 
The tree semantics for \QCTLsi is defined on labelled
$\Dirtreei[n]$-trees, and it is obtained by  enriching $\modelst$
as follows:
\[
\begin{array}{lcl}
  \ltree,\noeud\modelst \existsp{\obs} \phi &
  \mbox{if} & \mbox{there exists }\ltree'\Pequiv[p]\ltree \mbox{ such that
  }\ltree'\mbox{ is $\obs$-uniform in $p$ and }\ltree',\noeud\modelst\phi.
\end{array}
\]


Consider the following \CTL formula:
$\ligne{p}\egdef  \A \F p \wedge \A \G (p
\rightarrow \A\X \A\G \neg p)$.

This formula  holds in a labelled tree if and only if each path contains
exactly one node labelled with $p$. Therefore, evaluating the \QCTLi formula
$\existsp{\emptyset}\ligne{p}$ amounts to choosing a level of the
tree where to place one horizontal line
of $p$'s.
\bmcomment{Explain more?}


\section{Expressiveness}
\label{sec-expressivity}

In this section we study the expressive power of our logics. We first
observe that for both semantics, \QCTLsi and \QCTLi are equally
expressive. We then prove that with structure semantics  \QCTLi is
expressively equivalent to  \QCTL, and thus also to \MSO.
 Finally we
show that under tree semantics \QCTLsi is expressively equivalent to \MSO with equal level predicate.
Note that Theorem~\ref{th-qctli-qctl}, Corollary~\ref{cor-qctl-mso}
and Theorem~\ref{th-qctl-mso-eql} below only hold if the logics can
talk about the local states. For this reason, in this section
 we assume a set of dedicated atomic
propositions $\APlstates=\bigcup_{i\in
  [n]}\bigcup_{\lstate\in\setlstates_{i}}\{p_{l}\}\subset \AP$ such that
 for every \CKS $\CKS=(\setstates,\relation,\lab)$, for each $i\in [n]$ and
  $\lstate\in\setlstates_{i}$, for each state
  $\state=(\lstate_{1},\ldots,\lstate_{n})\in\setstates$, we have
  $p_{\lstate}\in\lab(\state)$ iff $\lstate_{i}=\lstate$.

\subsection{\QCTLsi, \QCTLi and \QCTL}
\label{sec-express1}

We first
remark that for the same reason why \QCTLs is no more expressive
than \QCTL, also \QCTLsi and \QCTLi are equally expressive (the proof of \cite[Proposition
  3.8]{DBLP:journals/corr/LaroussinieM14} readily applies):
\begin{theorem}
  \label{prop-expr-qctl-qctls}
  Under both semantics, \QCTLsi and \QCTLi are expressively equivalent.
\end{theorem}


We now
prove that for the structure semantics, \QCTLi is no more expressive
than \QCTL, and thus has the same expressivity as \MSO.

\begin{theorem}
  \label{th-qctli-qctl}
  Under structure semantics, \QCTLi and \QCTL are  expressively equivalent.
\end{theorem}

\begin{proof}
  It is quite clear that \QCTLi subsumes \QCTL. Observe however that
  the  quantifier on propositions  from \QCTL can be translated
  using the quantifier $\exists^{[n]}$ only because we have fixed the
  dimension of our models. If we allowed for
  models with arbitrary dimension we would have to add the classic
  quantifier $\exists$ in the syntax of \QCTLi for it to capture \QCTL.

  For the other direction, we define a translation
  $\transs{~~}$ from \QCTLi to \QCTL. We only provide the inductive case for
  the quantification on propositions, the others being trivial.
  \[
    \transs{\existsp{\obs} \phi} \egdef \existsp{}
    \left(\bigwedge_{(\lstate_{{1}},\ldots,\lstate_{{k}})\in
      \Dirtreei[{\obs\cap[n]}]}
    \A\G (\bigwedge_{i=1}^{k}p_{\lstate_{{i}}} \rightarrow
      p) \vee \A\G(\bigwedge_{i=1}^{k}p_{\lstate_{{i}}} \rightarrow
      \neg p )\right) \wedge\; \transs{\phi}.
  \]
 Observe that checking uniformity of $p$ in the
  reachable part of the model is sufficient, as the labelling of
  unreachable states is indifferent.
  It can be proven easily that
   for every \CKS $\CKS$, state $\state\in\CKS$ and  formula $\phi\in\QCTLi$, it
  holds that
  $\CKS,\state\modelss\phi \mbox{\;\; iff\;\; }\CKS,\state\modelss
  \transs{\phi}$.
%
%
\end{proof}

\begin{remark}
  \label{rem-size-transs}
One can check that $|\transs{\phi}|=O(n m^{n}|\phi|)$,
where $m=\max_{i\in [n]} |\setlstates_{i}|$.
\end{remark}

\subsection{\QCTLi and \MSO with equal level}
\label{sec-MSO}

We briefly recall the definition \MSOeql
(see, \eg, \cite{elgot-rabin66,thomas-msoeqlevel}
for more detail). In the following, $\setvarfo=\{x,y,\ldots\}$
(resp. $\setvarso=\{X,Y,\ldots\}$) is a countably-infinite set of
first-order (resp. second-order) variables. We also use a predicate
$P_{p}$  for each atomic
proposition $p\in\AP$.

The syntax of \MSO with equal level relation, or \MSOeql, is given by the following grammar:
\[\phi::= P_{p}(x) \mid x=y \mid \edge(x,y) \mid x\in X \mid \eql(x,y)
\mid \neg
\phi \mid \phi \ou \phi \mid \exists x. \phi \mid \exists X. \phi \]
where $p\in\AP$, $x,y\in\setvarfo$ and $X\in\setvarso$.

The syntax of \MSO is obtained by removing the $\eql(x,y)$ production
rule.
 We write
$\phi(x_{1},\ldots,x_{i},X_{1},\ldots,X_{j})$ to indicate that
variables $x_{1},\ldots,x_{i}$ and $X_{1},\ldots,X_{j}$ may appear
free in $\phi$. Without loss of generality we assume that a variable
cannot appear both free and quantified in a formula.
We use the standard semantics of \MSO,
the successor relation symbol $\edge$ being interpreted by the
transition relation on Kripke structures, and by the child relation on
trees. The semantics for \MSOeql is only defined on trees, and
$\eql(x,y)$ holds if the nodes denoted by $x$ and $y$ are at the same
depth in the tree. We write
$\model,\state_{1},\ldots,\state_{i},\setstates_{1},\ldots,\setstates_{j}\models
\phi(x_{1},\ldots,x_{i},X_{1},\ldots,X_{j})$ when $\phi$ holds in
model $\model$ when $x_{k}$ (resp. $X_{k}$) is interpreted as
$\state_{k}$ (resp. $\setstates_{k}$) for $k\in [i]$ (resp. $k\in [j]$). 

Since we aim at comparing the expressivity of \MSO (with equal level predicate in the case of tree
semantics) with that of the modal logic \QCTLi, we will consider \MSO formulas of the form $\phi(x)$,
where $x$ is a free variable representing the point of evaluation in
the model. 

First, we have seen that  under the structure semantics \QCTLi has
the same expressivity as \QCTL. Since \QCTL has the same expressivity
as \MSO (evaluated on reachable parts of the structures)
\cite{DBLP:journals/corr/LaroussinieM14}, we obtain the following corollary of Theorem~\ref{th-qctli-qctl}:

\begin{corollary}
  \label{cor-qctl-mso}
  Under structure semantics, \MSO and \QCTLi are equally expressive.
\end{corollary} 


We now turn to the case of the tree semantics. The constraint put on
the tree semantics for the proposition quantifier involves testing 
 length equality for arbitrarily long paths or, in terms of
trees, comparing the depths of arbitrarily deep nodes. It is thus not
a surprise that \QCTLi with tree semantics is more expressive than
\MSO on trees. It also seems natural that extending \MSO with the
equal level predicate allows to capture this constraint on 
proposition quantification, and thus that \MSOeql is as expressive
as \QCTLi with tree semantics. We establish with the following theorem that in fact the other
direction also holds.

\begin{theorem}
  \label{th-qctl-mso-eql}
  Under tree semantics, \MSOeql and \QCTLi are equally expressive.
\end{theorem}

\begin{proof}
We first show how to express in \MSOeql that two nodes in the unfolding
of a \CKS are $\obs$-indistinguishable (see
Definition~\ref{def-obs-trees}).
Let $\obs$ be an observation. We define the \MSOeql formula
$\phi^{\obs}(x,y)$ as follows:
\[\phi^{\obs}(x,y)\egdef \eql(x,y) \wedge \forall x'. \forall
y'. \left (x'\pref x \wedge y'\pref y \wedge \eql(x',y') \rightarrow
\bigwedge_{i\in \obs\cap [n]} \bigwedge_{\lstate\in\setlstates_{i}}
P_{p_{\lstate}}(x')\leftrightarrow P_{p_{\lstate}}(y')\right )\]
where $x'\pref x$ is an \MSO formula expressing that there is a path
from $x'$
to
  $x$.

Clearly, for every \CKS $\CKS$ and nodes $\noeud,\noeud'$ in its
unfolding $\unfold{\state}$ from some state $\state$, 
\[\unfold{\state},\noeud,\noeud'\models \phi^{\obs}(x,y) \mbox{ \;\;iff\;\;
}\noeud\oequivt\noeud'.\]
It is then easy to see that \QCTLi with tree semantics can be
translated into \MSOeql: the translation for \CTL is standard, and
propositional quantification with imperfect information can be expressed
using second order quantification and the above formula for
$\obs$-indistinguishability.

For the other direction, we build upon the following translation from \MSO to
\QCTL, presented in \cite{DBLP:journals/corr/LaroussinieM14}. 
For  $\phi(x,\liste{x}{1}{i},\liste{X}{1}{j})\in\MSO$, we inductively define
$\transt{\phi}$ as:
\[
\begin{array}{rclcrcl}
  \transt{P_{p}(x)} & = & p & \mbox{\hspace{2cm}} & \transt{P_{p}(x_k)} & = &
  \E\F (p_{x_k} \wedge p) \\[2pt]
  \transt{x=x_k} & = & p_{x_k} & & \transt{x_k=x_l} & = & \E\F(p_{x_k}\wedge
  p_{x_l}) \\[2pt]
\transt{x\in X_{k}} & = & p_{X_{k}} & & \transt{x_k\in X_{k}} & = & \E\F
(p_{x_k}\wedge p_{X_{k}})\\[2pt]
\transt{\neg \phi'} & = & \neg \transt{\phi'} & & \transt{\phi_{1}\vee
\phi_{2}} & = & \transt{\phi_{1}} \vee \transt{\phi_{2}}\\[2pt]
\transt{\exists x_k.\phi'} & = & \existsp[p_{x_k}]{}\uniq[p_{x_k}] \wedge
\transt{\phi'} & & \transt{\exists X_{k}.\phi'} & = & \existsp[p_{X_{k}}]{}
\transt{\phi'}\\[2pt]
  \transt{\edge(x,x_k)} & = & \E\X p_{x_k} & & \transt{\edge(x_k,x)} & = &
  \perp \\[2pt]
\multicolumn{7}{c}{\transt{\edge(x_k,x_l)} \;\;\; = \;\;\;  \E\F
  (p_{x_k}\wedge \E\X p_{x_l})}
\end{array}
\]
where $\uniq[p]\egdef \E\F p \wedge \forall q.\;\left ( \E\F(p \wedge
  q) \rightarrow \A\G (p \rightarrow q)\right )$ holds in a tree iff
it has exactly one node labelled with $p$. Observe that $x$ being
interpreted as the root of a tree 
it has no incoming edge, hence the translation of $\edge(x_k,x)$.  


We extend this translation into one from \MSOeql to \QCTLi
by adding the following rules:
\[
\begin{array}{rclcrcl}
  \transt{\eql(x,x_{k})} & = & p_{x_{k}} & \mbox{\hspace{1cm}} & \transt{\eql(x_{k},x_l)} & = &
  \existsp{\emptyset}\ligne{p} \wedge \A\G (p_{x_{k}}\rightarrow p \wedge p_{x_l}\rightarrow p)
\end{array}
\]
Observe that $x$ being interpreted as the root, $x_{k}$ is on the same
level as $x$ if and only if it is also assigned the root. Concerning the
second case, recall from Section~\ref{sec-QCTLi-tree-sem} that the
\QCTLi formula 
$\existsp{\emptyset}\ligne{p}$ places in the tree one unique horizontal line of
$p$'s. Requiring that $x_{k}$ and $x_{l}$ be both on this line thus ensures
that they are on the same level.
It is then easy to prove by induction the following lemma:
\begin{lemma}
  \label{lem-MSOeq-QCTLi}
For every  $\phi(x,x_{1},\ldots,x_{i},X_{1},\ldots,X_{j})\in\MSOeql$
and every pointed \CKS $(\CKS,\state)$,
 \[\unfold{\state},\state, \liste{\noeud}{1}{i},\liste{\setnodes}{1}{j}\models
\phi(x,\liste{x}{1}{i},\liste{X}{1}{j})$ \;\;iff\;\;
$\ltree'_{\KS}(\state),\state\modelst \transt{\phi}\]
where $\unfold{\state}'$ is obtained from $\unfold{\state}$ by
changing the labelling for variables $p_{x_{k}}$ and $p_{X_{k}}$ as
follows: $p_{x_{k}}\in\lab'(\noeud)$ if $\noeud=\noeud_{k}$ and
$p_{X_{k}}\in\lab'(\noeud)$ if $\noeud\in\setnodes_{k}$.
\end{lemma}
In particular, it follows that $\unfold{\state},\state\models \phi(x)$
iff $\unfold{\state},\state\models\transt{\phi}$.
\end{proof}

\begin{remark}
  \label{rem-normal-form}
  The two-way translation between \QCTLi and \MSOeql shows that when
  local states are identified by atomic propositions, there is a
  normal form for \QCTLi formulas involving only blind and
  perfect-information quantifiers.
\end{remark}
 


\section{Model checking \QCTLi}
\label{sec-modelchecking}

We now study the model-checking problem for \QCTLsi, both
for structure and tree semantics. In other terms, we study the problem
of deciding, given a finite \CKS $\CKS$, a state $\state\in\CKS$ and a
\QCTLsi formula $\phi$, whether it holds that $\CKS,\state\modelss\phi$
(or $\CKS,\state\modelst\phi$ for the tree semantics). 

\subsection{Structure semantics}
\label{sec-mc-struct}

We first prove that under structure semantics, similarly to \QCTLs and
\QCTL,  the model-checking problem is
\PSPACE-complete for both  \QCTLi and \QCTLsi. Observe that if  $n$ is
fixed the translation from
\QCTLi to \QCTL from Theorem~\ref{th-qctli-qctl}
suffices to obtain the upper bound.  But this translation, being exponential in $n$ (see
Remark~\ref{rem-size-transs}), is not enough if $n$ is not fixed; we provide an algorithm to show that the result holds even
if $n$ is part of the input.

\newcounter{mc-struct}
\setcounter{mc-struct}{\value{theorem}}
\begin{theorem}
  \label{mc-struct}
  Under structure semantics, model checking  \QCTLsi
 is \PSPACE-complete.
\end{theorem}

\begin{proof}
Hardness follows from the  \PSPACE-hardness of model checking \QCTL \cite{DBLP:journals/corr/LaroussinieM14}. 
For the upper bound, we modify the algorithm described in
\cite[Theorem 4.2]{DBLP:journals/corr/LaroussinieM14}. The
main difference is that when we guess a labelling for 
 $p$ on a \CKS $\CKS$,  we need to check that this labelling is
uniform.  With structure semantics this can be done in deterministic time
$O(|\CKS|^{2}\cdot n)$:
 look at every pair of states, and check that if they
are observationally equivalent (tested by comparing at most $n$
pairs of local states) then they agree on $p$.

We prove that the model-checking problem for \QCTLsi is in \PSPACE by
induction on the nesting depth $k$ of propositional quantification in
input formulas.
If $k=0$, \ie, the input formula is a \CTLs formula, call a
\CTLs model-checking algorithm running in polynomial space. For
nesting depth $k+1$,
the input formula $\phi$ is of the form $\phi=\Phi[q_{i}\mapsto
\existsp[p_{i}]{\obs_{i}}\phi_{i}]$, where  $\Phi$ is a  \CTLs
formula and for each $i$, $q_{i}$ is a fresh atomic proposition,
 $\obs_{i}$ is an observation and
 $\phi_{i}$ a \QCTLsi formula of nesting depth at most $k$. For each $i$, guess in linear time a
 labelling for $p_{i}$, check in quadratic time that it is uniform,
 evaluate formula $\phi_{i}$ in each state with this
 labelling, and mark states where it holds with $q_{i}$. By induction
 hypothesis, evaluating $\phi_{i}$ can be done in polynomial
 space. It just remains to evaluate the \CTLs formula $\Phi$ in polynomial
 space. The overall procedure thus runs in nondeterministic polynomial
 space, and because \NPSPACE = \PSPACE, the problem is in \PSPACE.
\end{proof}

\subsection{Tree semantics}
\label{sec-mc-tree}

We turn to the case of tree semantics. The first undecidability result
comes at no surprise since \QCTLi can  express the existence of
 winning strategies in imperfect-information games.
\begin{theorem}
    \label{theo-undecidable}
  Under tree semantics, the model-checking problem for \QCTLi is undecidable.
\end{theorem}
\begin{proof}
  The  \MSOeql theory of the binary tree is
  undecidable~\cite{lauchli1987monadic}, and with Lemma~\ref{lem-MSOeq-QCTLi} we obtain a
  reduction to the model-checking problem for \QCTLi.
\end{proof}

\subsubsection{Alternating tree automata}
\label{sec-ATA}

  We briefly recall the notion of alternating (parity) tree automata.
For a set $Z$, $\boolp(Z)$ 
is the set of
formulas built with elements of $Z$ as atomic propositions, using only connectives $\ou$ and
$\et$, 
and with $\top,\perp \in \boolp(Z)$.
An \emph{alternating tree automaton (\ATA) on $(\APf,\Dirtree)$-trees}
is a structure $\auto=(\tQ,\tdelta,\tq_{\init},\couleur)$
where 
$Q$ is a finite set of states, $\tq_{\init}\in \tQ$ is an initial
state, $\tdelta : \tQ\times 2^{\APf} \rightarrow \boolp(\Dirtree\times
\tQ)$ is a transition function, and $\couleur:\tQ\to \setn$ is a
colouring function.  To ease reading we shall write atoms in
transition formulas between brackets, such as
$[x,\tq]$.  
A \emph{nondeterministic tree automaton (\NTA) on
  $(\APf,\Dirtree)$-trees} is an \ATA
$\auto=(\tQ,\tdelta,\tq_{\init},\couleur)$ such that for every $\tq\in
\tQ$ and $a\in 2^{\APf}$, if $\tdelta(\tq,a)$ is written in
disjunctive normal form, then for every direction $\dir\in \Dirtree$,
each disjunct contains exactly one element of $\{\dir\}\times Q$.
The \emph{size} of an ATA is its number of states and its \emph{index}
is its number of different colours.

Because we work with trees that are not necessarily complete as
they represent unfoldings of Kripke structures, we
find it convenient to assume that the state set is partitioned
between $\tQp$ and $\tQm$: when sent in a direction where there is no
node in the input tree, states in $\tQp$ accept immediately while states in $\tQm$
reject immediately\footnote{Note that we could also work only with
  complete trees, with a special symbol labelling missing nodes.}.

  We also  recall the definition of acceptance by \ATA via games
between Eve and Adam. 
Let $\ATA=(\tQ,\tdelta,\tq_\init,\couleur)$ be an \ATA over $(\APf,\Dirtree)$-trees,
 let
$\ltree=(\tree,\lab)$ be such a tree and let $\noeud_\init \in \tree$.
 We
define the parity game $\tgame{\ATA}{\ltree}{\noeud_\init}=(\setpos,\moves,\pos_\init,\couleur')$: the set of
positions is $\setpos=\tree\times \tQ \times \boolp (\Dirtree\times \tQ)$, the
initial position is $\pos_{\init}=(\noeud_\init,\tq_\init,\tdelta(\tq_\init,\noeud_\init))$, and a position
$(\noeud,\tq,\pform)$ belongs to Eve if $\pform$ is of the form $\pform_1\vee
\pform_2$ or $[\dir,\tq']$; otherwise it belongs to
Adam. 
Moves in $\tgame{\tauto}{\ltree}{\noeud_\init}$ are defined by
the following rules:
\[\begin{array}{ll}
 (\noeud,\tq,\pform_1 \;\op\; \pform_2) \move (\noeud,\tq,\pform_i) &
 \mbox{where } 
 \op \in\{\vee,\wedge\} \mbox{ and } i\in\{1,2\},  \\
 \multicolumn{2}{l}{
     (\noeud,\tq,[\dir,\tq']) \move
     \begin{cases}
            (\noeud\cdot
            \dir,\tq',\tdelta(\tq',\lab(\noeud\cdot\dir))) & \mbox{if
            }\noeud\cdot\dir\in\ltree \\
            (\noeud,\tq,\top) & \mbox{if }\noeud\cdot\dir\notin\ltree
            \mbox{ and }\tq\in\tQ^{\top}\\
            (\noeud,\tq,\perp) & \mbox{if }\noeud\cdot\dir\notin\ltree
            \mbox{ and }\tq\in\tQ^{\perp}
     \end{cases}}
\end{array}\]


Positions of the form $(\noeud,\tq,\top)$ and  $(\noeud,\tq,\perp)$
are deadlocks, winning for Eve and Adam respectively.
 The colouring is inherited from the one of the automaton:
 $\couleur'(\noeud,\tq,\pform)=\couleur(\tq)$.

A tree $\ltree$ is \emph{accepted} from node $\noeud$  by
$\tauto$ if Eve has a winning strategy in
$\tgame{\tauto}{\ltree}{\noeud}$, and   we let $\lang(\auto)$ be the
 set of trees  accepted by $\auto$ from their root.

We recall three classic results on tree automata. The first one is
that nondeterministic tree automata are closed under projection, and
was established by
Rabin to deal with second-order monadic quantification:
\begin{theorem}[Projection \cite{rabin1969decidability}]
  \label{theo-projection}
  Given an \NTA $\NTA$ and an atomic
  proposition $p\in\AP$, one can build an \NTA $\proj{\NTA}$ 
 of same
  size and index such that
  $\lang(\proj{\NTA})=\proj{\lang(\NTA)}$. 
\end{theorem}

Because it will be important to understand the automata construction
for our decision procedure in Section~\ref{sec-decidable}, we briefly recall
that the projected automaton $\proj{\NTA}$ is simply automaton $\NTA$ 
with the only difference that when it reads the label of a node, it
can choose whether $p$ is there or not: if $\tdelta$ is the transition
function of $\NTA$, that of $\proj{\NTA}$ is
$\tdelta'(q,a)=\tdelta(q,a\union \{p\}) \ou
\tdelta(q,a\setminus\{p\})$, for any state $q$ and label $a\in 2^{\APf}$. Another way of seeing it is that
$\proj{\NTA}$ first guesses a $p$-labelling for the input tree, and
then simulates $\NTA$ on this modified input.
To prevent $\proj{\NTA}$ from guessing different labels for a same
node in different executions, it is crucial that $\NTA$ be nondeterministic, reason why we
need the next classic result: 
 the crucial simulation theorem, due to Muller and
Schupp.
\begin{theorem}[Simulation \cite{DBLP:journals/tcs/MullerS95}]
\label{theo-simulation}
Given an \ATA $\ATA$, one can build an \NTA $\NTA$ of exponential size
 and linear index such that $\lang(\NTA)=\lang(\ATA)$.
\end{theorem}

The last one was established by Kupferman and Vardi to deal with
imperfect information aspects in distributed synthesis. The rough idea is
that,  if one just observes
 $\Dirtree$,  uniform  $p$-labellings on 
$\Dirtree\times\Dirtreea$-trees  can be obtained by choosing the
labellings directly on $\Dirtree$-trees, and then lifting them to $\Dirtree\times\Dirtreea$.
\begin{theorem}[Narrowing \cite{kupferman1999church}]
  \label{theo-narrow}
  Given an \ATA $\ATA$ on $\Dirtree\times\Dirtreea$-trees,
  one can build an \ATA $\narrow[\Dirtree]{\ATA}$ on $\Dirtree$-trees
  of same size 
 such that for all $\dira\in\Dirtreea$, $\ltree\in\lang(\narrow[\Dirtree]{\ATA})$ iff $\liftI[\Dirtree\times\Dirtreea]{\dira}{\ltree}\in\lang(\ATA)$.
\end{theorem}

In fact the result in \cite{kupferman1999church} is stated for
$\ltree$ (and thus also $\liftI[\Dirtree\times\Dirtreea]{}{\ltree}$) a complete  tree, but the proof transfers
straightforwadly to this slightly more general result.

\subsubsection{A decidable fragment: hierarchy on observations}
\label{sec-decidable}

We turn to our main result,
which is the identification of an important decidable fragment.

\begin{definition}
  \label{def-hierarchical}
  A \QCTLsi formula $\phi$ is \emph{hierarchical} if for all
  subformulas $\phi_{1},\phi_{2}$ of the form
  $\phi_{1}=\existsp[p_{1}]{\obs_{1}}\phi'_{1}$ and
  $\phi_{2}=\existsp[p_{2}]{\obs_{2}}\phi'_{2}$ where  
  $\phi_{2}$
  is a subformula of $\phi'_{1}$, we have $\obs_{1}\subseteq\obs_{2}$.
\end{definition}

In other words, a formula is hierarchical if innermost propositional
quantifiers observe at least as much as  outermost ones.
We let \QCTLsih be the set of hierarchical \QCTLsi formulas.

\begin{theorem}
  \label{theo-decidable}
Under tree semantics, model checking \QCTLsih is non-elementary decidable.
\end{theorem}

In order to prove this we establish Lemma~\ref{lem-final} below, but
  we first introduce a few more notations.
For every $\phi\in\QCTLsi$, we let $\Iphi\egdef \biginter_{\obs\in\setobs}\obs$, where $\setobs$ is the set of observations
  that occur in $\phi$, with the intersection over the empty set 
  defined as $[n]$. We also let $\Dirtreei[\phi]\egdef
  \Dirtreei[\Iphi]$ (recall that for $I\subseteq [n]$,
  $\Dirtreei=\bigtimes_{i\in I}\setlstates_{i}$).
We will need a final important definition.

\begin{definition}[Merge]
  \label{def-merge}
Let
 $\ltree=(\tree,\lab)$ be an
$(\APf,\Dirtree)$-tree and  $\ltree'=(\tree',\lab')$ an
$(\APf\,',\Dirtree)$-tree. We
 define the \emph{merge} of $\ltree$ and $\ltree'$
 as the $(\APf\union\APf\,')$-labelled $\Dirtree$-tree $\ltree\merge\ltree'\egdef
(\tree\cap\tree',\lab'')$, where
$\lab''(\noeud)=\lab(\noeud) \union \lab'(\noeud)$.
\end{definition}

We explain the  idea behind this definition. In our decision procedure,
quantification on atomic propositions is performed by means of
automata projection (see Theorem~\ref{theo-projection}). But in order to obtain
uniform labellings for these propositions, we need 
to first narrow down our automata and our trees (see Theorem~\ref{theo-narrow}), and in this process we lose
information on the labelling of atomic propositions in the \CKS $\CKS$ on
which we evaluate the formula. 
To address this problem, first we assume without loss of generality that propositions that are quantified upon in $\Phi$
do not appear free in $\Phi$. We can then partition propositions in $\Phi$ between
those that are quantified upon, $\APq$, and those that
appear free, $\APfree$. We
 use the input tree 
 of the automaton we build to carry the labelling
for $\APq$, and in the end the
input tree is merged with the unfolding of $\CKS$ that carries
the labelling to evaluate propositions in $\APfree$.



\newcounter{lem-final}
\setcounter{lem-final}{\value{theorem}}
  \begin{lemma}
    \label{lem-final}
    Let $\Phi\in\QCTLsih$ 
with    $\APq$ and $\APfree$ defined as above, and let $\CKS$ be a finite \CKS over $\APfree$. For every subformula
    $\phi$ of $\Phi$ and state $\state$ of $\CKS$, one can build an
    \ATA $\bigauto[\state]{\phi}$ on $(\APq,\Dirtreei[\phi])$-trees
    such that for every 
    $(\APq,\Dirtreei[\phi])$-tree $\ltree$ rooted in $\projI[{\Dirtreei[\phi]}]{\state}$,
    \begin{equation*}
      \ltree\in\lang(\bigauto[\state]{\phi}) \mbox{\;\;\;iff\;\;\;}
      \liftI[{[n]}]{\dira}{\ltree}\merge\;\unfold{\state} \modelst \phi,
      \mbox{\;\;\; where }\dira=\projI[{[n]\setminus I_{\phi}}]{\state}.
    \end{equation*}

  \end{lemma}

For an $\Dirtreei[I]$-tree $\ltree$, from now on $\liftI[{[n]}]{}{\ltree}\merge
\;\unfold{\state}$ 
 stands for       $\liftI[{[n]}]{\dira}{\ltree}\merge
\;\unfold{\state}$, where $\dira=\projI[{[n]\setminus I}]{\state}$.

\begin{proof}[Proof] 
  Let $\Phi\in \QCTLsih$, and let
  $\APq$ (resp. $\APfree$) be the set of atomic propositions that are
  quantified upon (resp. that appear free) in
  $\Phi$. Modulo renaming of atomic propositions, we can assume without loss of generality that $\APq$ and
  $\APfree$ are disjoint.   Let $\CKS=(\setstates,\relation,\labS)$
  be a finite  \CKS over $\APfree$.
    For each state $\state\in\setstates$ and each
  subformula $\phi$ of $\Phi$ (note that all subformulas of $\Phi$ are also hierarchical), we define
  by induction on $\phi$ the \ATA
  $\bigauto{\phi}$. The definition builds upon the classic construction for \CTLs
  from~\cite{DBLP:journals/jacm/KupfermanVW00}.
  
\begin{description}

\item[$\phi=p$:] We let $\bigauto{p}$ be the \ATA over $\Dirtreei[{[n]}]$-trees with one unique
  state $\tq_\init$, with transition function defined as follows:
  \[\tdelta(\tq_\init,a)=
  \begin{cases}
    \top  & \mbox{if } (p\in\APfree \mbox{ and }p\in\labS(\state))  \mbox{ or
    }(p\in \APq \mbox{ and }p\in a)\\
    \perp & \mbox{if }  (p\in\APfree \mbox{ and }p\notin\labS(\state))  \mbox{ or
    } (p\in \APq \mbox{ and }p\notin a)
  \end{cases}
\]
The idea is that since we know the state $\state\in\setstates$ in
which we want to evaluate the formula, we can read the labelling for
atomic propositions in $\APfree$ (those that are not quantified upon)
directly from $\state$. However, for propositions in $\APq$, we need
to read them from the input tree. Indeed, if $p\in\APq$ it means that
$p$ is quantified upon in $\Phi$: there is a subformula
$\existsp{\obs} \phi$ of $\Phi$ such that $p$ is a subformula of
$\phi$. The automaton $\bigauto{\existsp{\obs}\phi}$ will be built by
narrowing, nondeterminising and projecting $\bigauto{\phi}$ on $p$. On
a given input tree $\ltree$, $\bigauto{\existsp{\obs}\phi}$ will thus
guess a labelling for $p$ in each node of $\ltree$ and simulate
(the nondeterminised narrowing of) $\bigauto{\phi}$ on this modified
input. $\bigauto{\phi}$ must therefore read the
labelling for $p$ from its input tree.

\item[$\phi=\neg \phi'$:] We obtain $\bigauto{\phi}$  by
  dualising  $\bigauto{\phi'}$, which is a classic operation on \ATA{}s.

\item[$\phi=\phi_1\ou\phi_2$:] Because
  $\Iphi=\Iphi[\phi_{1}]\cap\Iphi[\phi_{2}]$, and each
  $\bigauto{\phi_{i}}$ works on $\Dirtreei[\phi_{i}]$-trees, we need to
  narrow them so that they work on $\Dirtreei[\phi]$-trees: 
  for $i\in \{1,2\}$, we let $\ATA_{i}\egdef
  \narrow[\Iphi]{\bigauto{\phi_{i}}}$.
  
We then build $\bigauto{\phi}$ by taking the disjoint union of $\ATA_{1}$ and $\ATA_{2}$
  and adding a new initial state that nondeterministically chooses
  which of $\ATA_{1}$ or $\ATA_{2}$ to execute on the input tree, so
  that $\lang(\bigauto{\phi})=\lang(\ATA_{1})\union\lang(\ATA_{2})$.

\item[$\phi=\E\psi$:]
  The aim is to build an automaton $\bigauto{\phi}$ that works on
  $\Dirtreei[\phi]$-trees and that on input $\ltree$,  checks for the
  existence of a path in
  $\liftI[{[n]}]{}{\ltree}\merge\;\unfold{\state}$ that
  satisfies $\psi$.
  To do so,  $\bigauto{\phi}$  guesses a path $\tpath$ in $(\CKS,\state)$.
It  remembers the current state in $\CKS$, which provides the
  labelling for atomic propositions in $\APfree$, and while it guesses
  $\tpath$ it follows its projection on $\Dirtreei[\phi]$ in its input
  tree $\ltree$,
 reading the labels  to evaluate propositions in
  $\APq$.
  
Let $\max(\psi)=\{\phi_1,\ldots,\phi_n\}$ be the
  set of maximal state sub-formulas of $\psi$.
  In a first step we
  see these maximal state sub-formulas as atomic propositions. Formula
  $\psi$
  can thus be seen as an \LTL formula, and we can build 
  a nondeterministic
  parity word automaton
  $\autopsi=(\Qpsi,\Deltapsi,\qpsi_\init,\couleurpsi)$ over alphabet $2^{\max(\psi)}$
  that accepts
  exactly the models of $\psi$.  
  We define the
  \ATA
  $\tauto$  that, given as
  input a $(\max(\psi),\Dirtreei[\phi])$-tree $\ltree$,
  nondeterministically guesses a
path $\tpath$ in   $\liftI[{[n]}]{}{\ltree}\merge\;\unfold{\state}$  and
  simulates $\autopsi$ on it, assuming that the labels it reads
  while following $\projI[{\Dirtreei[\phi]}]{\tpath}$
  in
  its input correctly represent the truth value of formulas in
  $\max(\psi)$ along $\tpath$. 
Recall that $\CKS=(\setstates,\relation,\labS)$; we define
 $\tauto\egdef(\tQ,\tdelta,\tq_{\init},\couleur)$, where
\begin{itemize}
\item $\tQ=\Qpsi\times\setstates$, 
\item $\tq_{\init}=(\qpsi_{\init},\state)$,
\item $\couleur(\qpsi,\state')=\couleurpsi(\qpsi)$, and
\item  for each $(\qpsi,\state')\in\tQ$
  and $a\in 2^{\max(\psi)}$, 
  \[\tdelta((\qpsi,\state'),a)=\bigvee_{\tq'\in\Deltapsi(\qpsi,a)}\bigvee_{
    \state''\in\relation(\state')}[\projI[{\Dirtreei[\phi]}]{\state''},\left(\tq',\state''\right)].\]
\end{itemize}

The intuition is that $\tauto$ reads the current label, chooses nondeterministically
which transition to take in $\autopsi$, chooses a next state in $\CKS$
and proceeds in the corresponding direction in $\Dirtree_{\phi}$. 
To ensure\footnote{Actually this is not very important since the tree $\ltree$ on which our
  automata will work will always be such that the domain of
  $\liftI[{[n]}]{}{\ltree}$ contains the domain of $\unfold{\state}$.} that the path it guesses is not only in
$\unfold{\state}$ but also in
$\liftI[{[n]}]{}{\ltree}$, it is
enough to make sure that it always tries to stay inside its input tree $\ltree$,
which is achieved by letting
$\tQ^{\top}=\emptyset$ and $\tQ^{\perp}=\tQ$.
Thus, $\tauto$
 accepts exactly the $\max(\phi)$-\labeled $\Dirtreei[\phi]$-trees $\ltree$
 in which there exists a path that corresponds to some path in
 $\liftI[{[n]}]{}{\ltree}\merge\; \unfold{\state}$  that satisfies
 $\psi$, where maximal state formulas are considered as atomic propositions. 

Now from $\tauto$ we build the automaton $\bigauto{\phi}$ over
$\Dirtreei[\phi]$-trees labelled with real atomic propositions in
$\APq$.
In each node it visits, this automaton guesses what should be its
labelling over $\max(\psi)$, it simulates $\tauto$
accordingly, and checks
that the guesses it makes are correct.
If the path being guessed in $\liftI[{[n]}]{}{\ltree}\merge\;\unfold{\state}$
is currently in node $\noeud$ ending with state $\state'$, and
$\bigauto{\phi}$ guesses that $\phi_{i}$ holds in $\noeud$,
it checks this guess by starting a
 simulation of automaton $\bigauto[\state']{\phi_{i}}$ from node
 $\noeuda=\projI[{\Dirtreei[\phi]}]{\noeud}$ in its input $\ltree$.

For each $\state'\in\CKS$ and each $\phi_{i}\in\max(\psi)$ we first
build $\bigauto[\state']{\phi_i}$, which works on $\Dirtreei[\phi_i]$-trees. 
 Observe that $\Iphi[\phi]=\inter_{i=1}^n \Iphi[\phi_i]$, so that we need to
narrow down these automata:
We let $\ATA^i_{\state'}\egdef\narrow[{\Iphi[\phi]}]{\bigauto[\state']{\phi_i}}
=(\tQ^{i}_{\state'},\tdelta^{i}_{\state'},\tq^{i}_{\state'},\couleur^{i}_{\state'})$.
We also let
$\compl{\ATA^{i}_{\state'}}=(\compl{\tQ^{i}_{\state'}},\compl{\delta^{i}_{\state'}},\compl{\tq^{i}_{\state'}},\compl{\couleur^{i}_{\state'}})$
be its dualisation, and we assume w.l.o.g. that all the state sets are
pairwise disjoint. 
We define the \ATA $\bigauto{\phi}=(\tQ\cup
\bigcup_{i,\state'} \tQ^{i}_{\state'} \cup
\compl{\tQ^{i}_{\state'}},\tdelta',\tq_{\init},\couleur')$, where the
colours of states are left as they were in their original automaton,
and $\tdelta$ is defined as follows. For states in $\tQ^{i}_{\state'}$
(resp. $\compl{\tQ^{i}_{\state'}}$), $\tdelta$ agrees with $\tdelta^{i}_{\state'}$
(resp. $\compl{\delta^{i}_{\state'}}$), and for $(\qpsi,\state')\in \tQ$ and $a\in
2^{\APq}$ we let
\[\tdelta'((\qpsi,\state'),a)=\bigou_{a'\in 2^{\max(\psi)}} \left( \tdelta\left((\qpsi,\state'),a'\right)\et \biget_{\phi_i\in
a'}\tdelta^{i}_{\state'}(\tq^{i}_{\state'},a)\et \biget_{\phi_i\notin a'}\compl{\delta^{i}_{\state'}}(\compl{\tq^{i}_{\state'}},a)\right).\]

\item[$\phi=\existsp{\obs}\phi'$:] We
  build automaton $\bigauto{\phi'}$ that works on $\Dirtreei[\phi']$-trees;
because $\phi$ is hierarchical, we have that $\obs\subseteq I_{\phi'}$
and we can narrow down $\bigauto{\phi'}$ to work on $\Dirtreei[\obs]$-trees and obtain
$\ATA_{1}\egdef\narrow[{\Dirtreei[\obs]}]{\bigauto{\phi'}}$. By
Theorem~\ref{theo-simulation} we can
nondeterminise it to get $\ATA_{2}$, which by
Theorem~\ref{theo-projection} we can project with respect to
$p$, finally obtaining $\bigauto{\phi}\egdef \proj{\ATA_{2}}$.
\end{description}

We now prove by induction on $\phi$ that the construction is correct. 
In each case, we let $\ltree=(\tree,\lab)$ be an
$(\APq,\Dirtreei[\phi])$-tree rooted in $\projI[{\Dirtreei[\phi]}]{\state}$.
\begin{description}
\item[$\phi=p$:] First, note that $I_{p}=[n]$, so that $\ltree$ is
  rooted in $\projI[{\Dirtreei[p]}]{\state}=\state$. Let us
  consider first the case where
 $p\in\APfree$. By definition of $\bigauto{p}$, we have that
 $\ltree\in\lang(\bigauto{p})$ iff $p\in\labS(\state)$. On the other
 hand, by definition of the merge operation, of the unfolding, and because $\APq$ and $\APfree$ are disjoint, we have
 $\liftI[{[n]}]{}{\ltree}\merge\;\unfold{\state}\models p$ iff
 $p\in\labS(\state)$, and we are done. Now if $p\in\APq$: by
 definition of $\bigauto{p}$, we have $\ltree\in\lang(\bigauto{p})$
 iff $p\in \lab(\state)$; also, by definition of the merge, 
 we have that
 $\liftI[{[n]}]{}{\ltree}\merge\;\unfold{\state}\models p$ iff
 $p\in\lab(\state)$, which concludes.
  \item[$\phi=\neg\phi'$:] trivial.
\item[$\phi=\phi_{1}\ou \phi_{2}$:] We have $\ATA_{i} =
\narrow[{\Dirtreei[\phi]}]{\bigauto{\phi_{i}}}$, so by Theorem~\ref{theo-narrow}
we get that $\ltree\in\lang(\ATA_{i})$ iff
$\liftI[{\Dirtreei[\phi_{i}]}]{}{\ltree}\in\lang(\bigauto{\phi_{i}})$, which by
induction hypothesis holds iff
$\liftI[{[n]}]{}{(\liftI[{\Dirtreei[\phi_{i}]}]{}{\ltree})}\merge\;
\unfold{\state} \modelst \phi_{i}$, \ie, iff
$\liftI[{[n]}]{}{\ltree}\merge\;\unfold{\state}\modelst\phi_{i}$.
We conclude by reminding that
$\lang(\bigauto{\phi})=\lang(\ATA_{1})\union\lang(\ATA_{2})$.

\item[$\phi=\E\psi$:] Suppose that $\ltree'=\liftI[{[n]}]{}{\ltree}\merge\;\unfold{\state}\modelst\E\psi$. There exists a
    path  $\tpath$ starting at the root $\state$ of $\ltree'$ such that $\ltree',\tpath\models\psi$. 
    Again, let $\max(\psi)$ be the
  set of maximal state subformulas of $\phi$, and let
    $w$ be the infinite word over $2^{\max(\psi)}$ that agrees with n
    $\tpath$ on the state formulas in $\max(\psi)$. By definition,
    $\autopsi$ has an accepting execution  on $w$. Now
    in the acceptance game of $\bigauto{\phi}$ on $\ltree$, Eve can guess the
     path $\tpath$, following $\projI[\Dirtree_{\phi}]{\tpath}$ in its
     input $\ltree$, and she can also guess the corresponding word $w$
     on $2^{\max(\psi)}$ and
an accepting execution of $\autopsi$ on $w$. Let
     $\noeud'\in\ltree'$ be a node of $\tpath$, $\state'$ its last direction and
let     $\noeud=\projI[{\Dirtreei[\phi]}]{\noeud'}\in\ltree$. Assume that in node
     $\noeud$ of the input tree, in a state $(\qpsi,\state')\in \tQ$, Adam challenges Eve on some
    $\phi_i\in\max(\psi)$ that she assumes to be true in $\noeud'$, \ie, Adam chooses 
    the conjunct $\tdelta^{i}_{\state'}(\tq^{i}_{\state'},a)$, where $a$
    is the label of $\noeud$. Note that in the evaluation game
    this means that Adam moves to position
    $(\noeud,(\qpsi,\state'),\tdelta^{i}_{\state'}(\tq^{i}_{\state'},a))$.
    We  want to show that Eve wins from this
    position.

    Let $\ltree_{\noeud}$ (resp. $\ltree'_{\noeud'}$) be the subtree
    of $\ltree$ (resp. $\ltree'$) starting in $\noeud$
    (resp. $\noeud'$)\footnote{If $\noeud=w\cdot \dir$, the subtree $\ltree_{\noeud}$ of $\ltree=(\tree,\lab)$ is defined
      as $\ltree_{\noeud}\egdef(\tree_{\noeud},\lab_{\noeud})$ with
      $\tree_{\noeud}=\{\dir\cdot w' \mid w\cdot\dir\cdot w' \in
      \tree\}$, and $\lab_{\noeud}(\dir\cdot w')=\lab(w\cdot\dir\cdot w')$:
      we  remove from each node all directions before $\last(\noeud)$.}. It is enough to show that $\ltree_{\noeud}$ is
    accepted by $\ATA^{i}_{\state'}=\narrow[I_{\phi}]{\bigauto[\state']{\phi_{i}}}$. 
Observe that $\ltree_{\noeud}$ is rooted in the last direction
of $\noeud=\projI[{\Iphi[\phi]}]{\noeud'}$, and since the last
direction of $\noeud'$ is $\state'$ we have that $\ltree_{\noeud}$ is
rooted in $\projI[{\Iphi[\phi]}]{\state'}$.
Let us write
$\ltree''=\liftI[{\Iphi[\phi_i]}]{\state''}{\ltree_{\noeud}}$, where
$\state''=\projI[{\Iphi[\phi_i]}]{\state'}$. 
By Theorem~\ref{theo-narrow}, because
$\ATA^{i}_{\state'}=\narrow[I_{\phi}]{\bigauto[\state']{\phi_{i}}}$
and  $\projI[{\Iphi[\phi]}]{\state'}=\projI[{\Iphi[\phi]}]{(\projI[{\Iphi[\phi_i]}]{\state'})}$,
we have that
\begin{equation}
  \label{eq:3a}
  \ltree_{\noeud}\in\lang(\ATA^i_{\state'}) \mbox{\bigiff} \ltree''\in\lang(\bigauto[\state']{\phi_{i}}).
\end{equation}


Since $\ltree''$ is rooted in    $\projI[{\Iphi[\phi_{i}]}]{\state'}$
we can apply   the induction 
hypothesis on $\ltree''$ with $\phi_{i}$, and we get that
\begin{equation}
  \label{eq:3}
\ltree''\in\lang(\bigauto[\state']{\phi_{i}}) \mbox{\bigiff}
\liftI[{[n]}]{}{\ltree''}\merge\; \unfold{\state'}\modelst \phi_{i}.  
\end{equation}
Now, because $\noeud'$ ends in $\state'$ we also have that
\begin{equation}
  \label{eq:4}
\ltree'_{\noeud'}=\liftI[{[n]}]{}{\ltree''}\merge\; \unfold{\state'}.  
\end{equation}
Putting \eqref{eq:3a}, \eqref{eq:3} and \eqref{eq:4} together, we
obtain that
\begin{equation}
  \label{eq:4bis}
\ltree_{\noeud}\in\lang(\ATA^{i}_{\state'}) \quad\mbox{iff}\quad\ltree'_{\noeud'}\modelst\phi_{i}.
\end{equation}
Because we have assumed that Eve guesses $w$ correctly, we also have
    that $\ltree',\noeud'\modelst\phi_i$, \ie,
    $\ltree'_{\noeud'}\modelst\phi_{i}$. This, together with
    \eqref{eq:4bis}, gives us
 that 
 $\ltree_{\noeud}$
 is accepted by
 $\ATA^{i}_{\state'}$.

 Eve thus has a winning strategy from the initial position of
    the acceptance game of  $\ATA^{i}_{\state'}$ on $\ltree_{\noeud}$.
    This initial position is
    $(\noeud,\tq^i_{\state'},\tdelta^i_{\state'}(\tq^i_{\state'},a))$. Since
    $(\noeud,\tq^i_{\state'},\tdelta^i_{\state'}(\tq^i_{\state'},a))$ and
    $(\noeud,(\qpsi,\state'),\tdelta^i_{\state'}(\tq^i_{\state'},a))$ contain the same node $\noeud$ and
    transition formula $\tdelta^i_{\state'}(\tq^i_{\state'},a)$, a winning
    strategy in one of these
    positions\footnote{Recall that positional strategies are
      sufficient in parity games \cite{DBLP:journals/tcs/Zielonka98}.} is also a winning strategy in the other, and therefore
    Eve  wins Adam's challenge.
With a similar argument, we get that also when Adam challenges Eve on
    some $\phi_i$ assumed not to be true in  node $\noeuda$, Eve wins
    the challenge. Finally,
    Eve wins the acceptance game of
    $\bigauto{\phi}$ on $\ltree$, and thus $\ltree\in\lang(\bigauto{\phi})$.

    For the other direction, assume that
    $\ltree\in\lang(\bigauto{\phi})$, \ie, Eve wins the evaluation game
    of $\bigauto{\phi}$ on $\ltree$. Again, let
    $\ltree'=\liftI[{[n]}]{}{\ltree}\merge\;\unfold{\state}$. A winning strategy for Eve
    describes a path $\tpath$ 
    in $\unfold{\state}$, which is also
 a path in $\ltree'$. This
    winning strategy also defines an infinite word $w$
    over $2^{\max(\psi)}$ such that $w$ agrees with $\tpath$ on the
    formulas in $\max(\psi)$, and it also describes an accepting run
    of $\autopsi$ on $w$. Hence $\ltree',\tpath\modelst\psi$, and
    $\ltree'\modelst \phi$.

  \item[$\phi=\existsp{\obs}\phi'$:] First, observe that because
    $\phi$ is hierarchical, we have that $I_{\phi}=\obs$. Next, by
    Theorem~\ref{theo-projection} we have that
    \begin{equation}
      \label{eq:5}
\ltree\in\lang(\bigauto{\phi}) \mbox{\bigiff there exists }
\ltree_{p}\Pequiv \ltree \mbox{ such that }
    \ltree_{p}\in\lang(\ATA_{2}).        
    \end{equation}
By Theorem~\ref{theo-simulation},
    $\lang(\ATA_{2})=\lang(\ATA_{1})$, and since $\ATA_{1}
    =\narrow[{\Dirtreei[\obs]}]{\bigauto{\phi'}}=\narrow[{\Dirtreei[\phi]}]{\bigauto{\phi'}}$
    we get by Theorem~\ref{theo-narrow} that
    \begin{equation}
      \label{eq:6}
      \ltree_{p}\in\lang(\ATA_{2}) \mbox{\bigiff}
      \liftI[{\Dirtreei[\phi']}]{\dira}{\ltree_{p}}\in\lang(\bigauto{\phi'}), \mbox{    where $\dira=\projI[(I_{\phi'}\setminus I_{\phi})]{\state}$.}   
    \end{equation}
 Now
    $\ltree_{p}$ and $\ltree$ have the same root,
    $\projI[{\Dirtreei[\phi]}]{\state}$.  The root of
    $\liftI[{\Dirtreei[\phi']}]{\dira}{\ltree_{p}}$ is thus
    $(\projI[{\Dirtreei[\phi]}]{\state},\dira)=\projI[{\Dirtreei[\phi']}]{\state}$,
    and we can
 apply 
the
induction hypothesis on
$\liftI[{\Dirtreei[\phi']}]{\dira}{\ltree_{p}}$ with $\phi'$: 
\begin{equation}
  \label{eq:7}
\liftI[{\Dirtreei[\phi']}]{\dira}{\ltree_{p}}\in\lang(\bigauto{\phi'})
\mbox{\bigiff} \liftI[{[n]}]{}{\liftI[{\Dirtreei[\phi']}]{\dira}{\ltree_{p}}\,}\merge\;\unfold{\state}\modelst
\phi'.  
\end{equation}
Now, with \eqref{eq:5}, \eqref{eq:6} and \eqref{eq:7} together with
the fact that
$\liftI[{[n]}]{}{\liftI[{\Dirtreei[\phi']}]{\dira}{\ltree_{p}}\,}\,=\liftI[{[n]}]{}{\ltree_{p}}$,
we get that
\begin{equation}
  \label{eq:8}
  \ltree\in\lang(\bigauto{\phi}) \mbox{\bigiff there exists }
  \ltree_{p}\Pequiv \ltree \mbox{ such that }
  \liftI[{[n]}]{}{\ltree_{p}}\merge\;\unfold{\state}\modelst\phi'.
\end{equation}
Let us prove that the right-hand side of \eqref{eq:8} holds if and
only if
$\liftI[{[n]}]{}{\ltree}\merge\;\unfold{\state}\modelst\existsp{\obs}\phi'$.
For the first direction, assume that there exists $\ltree_{p}\Pequiv
\ltree \mbox{ such that }
\liftI[{[n]}]{}{\ltree_{p}}\merge\;\unfold{\state}\modelst\phi'$.
First, by definition of the merge, because $p\in\APq$ and
$\APq$ and $\APfree$ are disjoint, the
$p$-labelling of $\liftI[{[n]}]{}{\ltree_{p}}\merge\;\unfold{\state}$
is determined by the $p$-labelling of $\liftI[{[n]}]{}{\ltree_{p}}$,
which by definition of the lift is $\obs$-uniform. In addition it is
clear that
$\liftI[{[n]}]{}{\ltree_{p}}\merge\;\unfold{\state}\Pequiv\liftI[{[n]}]{}{\ltree}\merge\;\unfold{\state}$,
which concludes this direction.

For the other direction, assume that
$\liftI[{[n]}]{}{\ltree}\merge\;\unfold{\state}\modelst\existsp{\obs}\phi'$:
there exists $\ltree'_{p}\Pequiv
\liftI[{[n]}]{}{\ltree}\merge\;\unfold{\state}$ such that
$\ltree'_{p}$ is $\obs$-uniform in $p$ and
$\ltree'_{p}\modelst\phi'$. Let us write
$\ltree'_{p}=(\tree',\lab'_{p})$ and $\ltree=(\tree,\lab)$. We define
$\ltree_{p}\egdef(\tree,\lab_{p})$ where for each $\noeud\in\tree$,
if there exists $\noeud'\in\tree'$ such that
$\projI[\obs]{\noeud'}=\noeud$, we let
\[\lab_{p}(\noeud)=
\begin{cases}
  \lab(\noeud)\union\{p\} & \mbox{if }p\in \lab'_{p}(\noeud')\\
  \lab(\noeud)\setminus\{p\} & \mbox{otherwise}.
\end{cases}
\]
This is well defined because $\ltree'_{p}$ is $\obs$-uniform in $p$:
if two nodes $\noeud',\noeuda'$ project on $\noeud$, we have
$\noeud'\oequivt\noeuda'$ and thus they agree on $p$.
In case there is no $\noeud'\in\tree'$ such that
$\projI[{\Dirtreei[\phi]}]{\noeud'}=\noeud$, we can let
$\lab_{p}(\noeud)=\lab(\noeud)$ as this node disappears in $\liftI[{[n]}]{}{\ltree}\merge\;\unfold{\state}$.
Clearly, $\ltree_{p}\Pequiv\ltree$. Now we write  
$\ltree''_{p}=\liftI[{[n]}]{}{\ltree_{p}}\merge\;\unfold{\state}$ and
we prove that $\ltree''_{p}=\ltree'_{p}$ hence $\ltree''_{p}\modelst\phi'$, which concludes.
It is clear that $\ltree''_{p}$ and $\ltree'_{p}$ have the same domain.
Also, because
$\ltree'_{p}\Pequiv\liftI[{[n]}]{}{\ltree}\merge\;\unfold{\state}$ and
$\ltree''_{p}=\liftI[{[n]}]{}{\ltree_{p}}\merge\;\unfold{\state}$, 
by definition of the merge both agree with $\unfold{\state}$ for all
atomic propositions in $\APfree$. Because $\ltree_{p}\Pequiv\ltree$,
and again by definition of the merge, $\ltree''_{p}$ and $\ltree'_{p}$
also agree on all atomic propositions in $\APq\setminus\{p\}$. Finally, by
definition of $\ltree_{p}$ and because $\ltree_{p}'$ is
$\obs$-uniform in $p$, we get that $\ltree''_{p}$ and
$\ltree'_{p}$  also agree on $p$,  and therefore $\ltree''_{p}=\ltree'_{p}$.
\end{description}
\vspace{-2em}
%
\end{proof}

We can now prove Theorem~\ref{theo-decidable}. Let
$(\CKS,\state)$ be a pointed \CKS, and let $\phi\in\QCTLsih$. By Lemma~\ref{lem-final} one can build an \ATA
$\bigauto{\phi}$ such that for every
 labelled $\Dirtreei[\phi]$-tree $\ltree$ rooted in
 $\projI[{\Dirtreei[\phi]}]{\state}$, it holds that
$ \ltree\in\lang(\bigauto{\phi}) \mbox{ iff }
      \liftI[{[n]}]{}{\ltree}\merge \;\unfold{\state} \modelst
      \phi$.
      Let $\tree$ be the full $\Dirtreei[\phi]$-tree rooted in
      $\projI[{\Dirtreei[\phi]}]{\state}$, and let
      $\ltree=(\tree,\lab_{\emptyset})$,  where $\lab_{\emptyset}$ is
       the empty labelling.
      Clearly,
$\liftI[{[n]}]{}{\ltree}\merge
\;\unfold{\state}=\unfold{\state}$, and because $\ltree$ is
rooted in $\projI[{\Dirtreei[\phi]}]{\state}$, we have
      $ \ltree\in\lang(\bigauto{\phi}) \mbox{ iff }
      \unfold{\state}\modelst \phi$. It only remains to build a simple
       deterministic tree automaton
      $\ATA$ over $\Dirtreei[\phi]$-trees
      such that $\lang(\ATA)=\{\ltree\}$, and check for emptiness of
      the alternating tree automaton
      $\lang(\ATA\cap\bigauto{\phi})$. Because 
      nondeterminisation makes the size of the automaton gain
      one exponential for each nested quantifier on propositions, the
      procedure is nonelementary, and hardness is inherited from the
      model-checking problem for \QCTL~\cite{DBLP:journals/corr/LaroussinieM14}. 


\section{Conclusion and future work}
\label{sec-conclusion}

We have introduced the essence  of imperfect information in \QCTLs, by
adding internal structure to  states of the models
and parameterising propositional quantifiers with observational power
over this internal structure.
We considered both the structure
and tree semantics,  intimately related to the notions of
\emph{no memory} and \emph{perfect recall} in game strategies,
respectively. For the structure semantics we showed that our logic
coincides with \QCTL in expressive power, and thus also with \MSO, and
that the model-checking problem is \PSPACE-complete, as for \QCTL.
For the tree semantics however we showed that our logic is
expressively equivalent to \MSO with equal level, and that its
model-checking problem is thus undecidable. But we established, thanks
to automata techniques made possible by our modelling choices, that
model checking hierarchical formulas is decidable.

  
Several future work
directions await us. 
First it would be interesting to study \QCTLi under the amorphous
semantics, studied by French for \QCTL 
in~\cite{french2001decidability}.
We would also like to investigate fragments with better complexity, as
well as the satisfiability problem for \QCTLi. Then we believe that
there may be interesting connections with Chain Logic with equal
level, a restriction of \MSOeql that is decidable on trees. Does it
correspond to another interesting decidable fragment of \QCTLsi?
Finally, we aim at exploiting our last result in
various logics for strategic reasoning with imperfect information,
such as \ATLSsc and \SL. 

%
%
%




\end{document}